\documentclass[letterpaper, 10 pt, conference]{ieeeconf}
\IEEEoverridecommandlockouts 
\usepackage[utf8]{inputenc} 
\usepackage[T1]{fontenc}
\usepackage{url}
\usepackage{ifthen}
\usepackage{cite,color}
\usepackage[cmex10]{amsmath}

\usepackage[utf8]{inputenc}
\usepackage{geometry}
\usepackage{graphicx}
\usepackage{amsmath}
\usepackage{amssymb}
\usepackage{amsthm}
\usepackage{subfigure}
\usepackage{mathtools}
\usepackage{eucal}

\title{\LARGE \bf Redundancy Scheduling in Systems with Bi-Modal \\Job Service Time Distributions} 

\author{Amir Behrouzi-Far and Emina Soljanin$^{1}$
\thanks{$^{1}$The authors are with the Department of Electrical and Computer Engineering,
        Rutgers University, the State University of New Jersey, USA
        {\tt\small \{amir.behrouzifar,emina.soljanin\}@rutgers.edu}.}%
}

\newtheorem{theorem}{Theorem}
\newtheorem{definition}{Definition}
\newtheorem{lemma}{Lemma}
\newtheorem{proposition}{Proposition}

\begin{document}

\maketitle

\begin{abstract}
Queuing systems with redundant requests have drawn great attention because of their promise to reduce the job completion time and variability. Despite a large body of work on the topic, we are still far from fully understanding the benefits of redundancy in practice. We here take one step towards practical systems by studying queuing systems with bi-modal job service time distribution. Such distributions have been observed in practice, as can be seen in, e.g., Google cluster traces.
We develop an analogy to a classical urns and balls problem, and use it to study the queuing time performance of two non-adaptive classical scheduling policies: random and round-robin. We introduce new performance indicators in the analogous model, and argue that they are good predictors of the queuing time in non-adaptive scheduling policies. We then propose a non-adaptive scheduling policy that is based on combinatorial designs, and show that it has better performance indicators. Simulations confirm that the proposed scheduling policy, as the performance indicators suggest, reduces the queuing times compared to random and round-robin scheduling. 
\end{abstract}{}

\section{Introduction}
In distributed computing/storage systems, redundancy plays an important role in harnessing the mean sojourn time of jobs, \cite{dean2013tail} and \cite{joshi2014delay}. It is widely studied in theory \cite{joshi2017efficient,lee2017speeding}, and employed in practice, \cite{he2010comet,bernardin2006using}. Redundancy is particularly beneficial in systems with high variability in job service times \cite{gardner2017redundancy}.

In systems with multiple servers, each with its own queue, submitting redundant copies of an arriving job to multiple servers brings twofold benefit. First, by waiting in multiple queues, jobs get to service faster (on average) than the no-redundancy scenario. This happens because redundancy brings \textit{diversity} to queuing time of the jobs, which on average decreases the waiting time in the queue since some copies of each job face shorter queues than the others. The second benefit arises when jobs experience high variability in their service time. This problem, which is caused by "stragglers" \cite{ananthanarayanan2013effective}, can significantly increase the average service time and its variability \cite{aktas2019straggler}. By letting jobs be served concurrently at multiple servers, jobs experience the minimum service time across the servers \cite{behrouzi2018effect}. 

Two scenarios have been proposed for treating redundancy in distributed systems. In the first scenario, \cite{joshi2017efficient} and \cite{raaijmakers2018delta}, copies of an arriving job are submitted to multiple servers and the redundant copies get cancelled once the first copy \textit{starts} service. The first copy of a job starting the service is the one which faces the least-work-left queue among all the copies. This scenario is beneficial when the variability in queuing times is high. For instance, systems with bi-modal or multi-modal service time of jobs could benefit from it. Note that, in this scenario redundancy incurs no extra service cost, since only one copy of a job gets into service. In a second scenario, \cite{gardner2017redundancy}, on the other hand, copies of a job may get into service concurrently and the redundant copies get cancelled only when the first copy \textit{finishes} service. The gain of this scenario is that it can reduce the variability in the life time of the jobs, i.e, queuing time plus service time. As an example, systems with heterogeneous servers or homogeneous servers with non-careful resource sharing could benefit from this scenario.

Job scheduling problems in distributed systems has been under study since the emergence of these systems \cite{wang1985load}. This problem has been studied with different objectives, e.g., load balancing \cite{harchol1996exploiting}, fairness \cite{bansal2001analysis}, minimizing the average/variability of job life time in the system \cite{wierman2005classifying}, maximizing resource utilization \cite{peng2015r}. Since these objectives are not always aligned, studies have been devoted to designing scheduling policies that can provide a reasonable trade-off \cite{bansal2001analysis}. In the case of jobs with redundancy, a related line of work looks into how many servers should be used for job replication, \cite{gardner2017redundancy} and \cite{aktas2019learning}. In general, scheduling policies can be categorized into \textit{adaptive} and \textit{non-adaptive} policies. While adaptive policies, e.g., join the shortest queue \cite{gupta2007analysis} or power-of-d choices \cite{gardner2017redundancy}, use information of queues' length/servers' workload to pick servers for arriving jobs, non-adaptive policies, e.g., random or round robin \cite{harchol1998choosing}, make a blind decision. In systems with redundancy, multiple servers have to be selected, and thus the problem could be exponentially harder in a system with redundancy than in its no-redundancy counterpart. 
The role of redundancy to \textit{diversify} the queuing/service time also adds a new dimension to the scheduling problem, and the classical scheduling policies, designed for no-redundancy systems, may not have the expected performance in systems with redundancy.

In this work, we consider a distributed system with $n$ servers, each with its own queue. Job service times follow a bi-modal behavior: an arriving job (on average) has either a \textit{short} or \textit{long} service time. The service requirement of a job is not known upon arrival. This model has been observed in practice; see e.g, Google trace data \cite{chen2010analysis}. Upon arrival, $r$ copies of each job $(r\leq n)$ get scheduled in $r$ different servers. The service time of each copy is sampled from a \textit{fast} exponential distribution if the job is short or a \textit{slow} exponential distribution if the job is long. With bi-modal jobs' service time, there will be a high queuing time variability across the copies of a job. Nevertheless, since the exponential distribution is not heavy-tailed, a job will not face high variability in service time across its redundant copies. Accordingly, we consider the first scenario for handling redundancy, where the redundant copies of a job get cancelled as soon as the first copy starts the service. In this setup, by developing an analogy to the classical urns and balls problem, we study the performance of two well known non-adaptive scheduling policies, random and round-robin. We then propose a new non-adaptive scheduling policy based on combinatorial block designs which shares the positive capabilities of both random and round-robin policies without having their negative properties. Then, through simulation of the queuing system, we show that the new scheduling policy outperforms the existing non-adaptive policies by a considerable margin.

This paper is organized as follows. In Sec.~\ref{sysModel_probStat}, we present the queuing system model. In Sec.~\ref{urnsBallsAnalog}, we develop an urns and balls analogy to our queuing problem, and define and analyze there performance indicators in the analogous system. In Sec.~\ref{simResult}, we analyse our queuing system with different scheduling policies by simulation, and show that its performance is indeed well predicted by the performance indicators of the analogous urns and balls system.

\section{System Model and Problem Statement}\label{sysModel_probStat}

\subsection{System Architecture and Service model}
In our studied system, shown in Fig. \ref{fig:sysModel},  there are $n$ identical servers. Upon arrival, $r$ copies of a job get scheduled into $r$ servers and once the first copy enters service the redundant copies get cancelled. The service time of each copy of a job is sampled from a fast exponential distribution if the job is short and from a slow exponential distribution if the job is long. Let's define the exponential random variable $\tau$, with rate $\mu_1$, to be the service time of short jobs and
    \begin{equation}
        q=\frac{\textup{average service time of long jobs}}{\textup{average service time of short jobs}}.
        \label{q}
    \end{equation}{}
The average service time of long jobs is the random variable $q\tau$, which is an exponential with rate $\mu_1/q$. This model for scaling the service time distribution is also proposed in \cite{gardner2017better}. We model the frequency of arriving a long job by $p\in[0,1]$, such that an arriving job is a long one with probability $p$ and a short one with probability $1-p$.

    \begin{figure}[t]
        \centering
        \includegraphics[width=7.8cm]{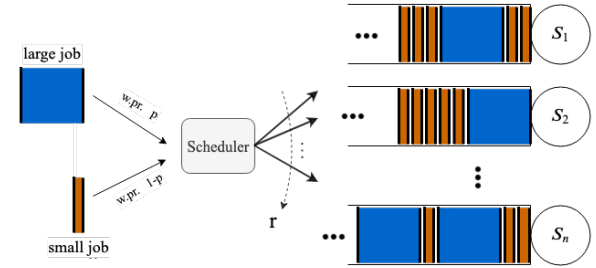}
        \caption{System model.}
        \label{fig:sysModel}
    \end{figure}
\subsection{Scheduling Policies}
We consider non-adaptive scheduling policies, where servers are selected for arriving jobs with no side knowledge of the workloads on servers. We are interested in the effect of scheduling policies on the queuing time of jobs, defined as the waiting time in the queue experienced by the first copy that starts service. Among non-adaptive policies, we specially focus on random and round-robin, due to their popularity both in theory and practice. With random policy, for each arriving job $r$ servers are randomly selected and one copy gets assigned to each server. Round-robin policy on the other hand takes a more structured approach, where it assigns $j$th copy of $i$th arriving job to the server indexed by $((i-1)r+j)\mod n$ for $j=1,2,\dots,r$. We then propose a new scheduling policy based on combinatorial block designs, which is also a non-adaptive policy, and show how it reduces the average queuing time by managing the \textit{overlaps} between set of selected servers across the jobs.

For scheduling jobs with bi-modal service time redundantly, the following two considerations have to be taken into account. First, load balancing should not be compromised, in the sense that the amount of actual workload assigned to the servers should be balanced. Note that, load balancing has been the central objective for designing schedulers in queuing systems. Nevertheless, we have observed that, with bi-modal jobs' service time behavior, the following objective should also be considered while designing a scheduling policy: \textit{the set of servers assigned to a job should make as less overlaps as possible with the set of servers assigned to the existing jobs in the system}. To understand the importance of the second objective, consider the case where a small job arrives and, among others, there is a large job waiting in $r$ queues. If the set of servers assigned to the small job completely overlaps with the servers that are already assigned to the large job then the small job is trapped behind the large one in every queue its assigned to. Therefore, it will definitely experience a long service time while waiting in the queue. Whereas, if the sets of assigned servers has a comparatively few overlaps, then the small job may find a queue in which it has to wait much less than the queues in which it has overlap with a large job. Therefore, with bi-modal jobs' service time behavior, a scheduler should be able to minimize overlaps between the set of assigned servers to jobs. In other words, it should be able to provide enough diversity in queuing times across the redundant copies of every arriving job.

To improve load balancing, a round-robin assignment of jobs may be helpful. This policy achieves the best expected balance in the loads assigned to servers. However, in terms of diversity of redundancy it is inferior to the other policies we study here, the reason of which will be quantified in Section \ref{urnsBallsAnalog}. Random policy, on the other hand, performs better in terms of diversifying redundancy but it is not as effective as round-robin in terms of load balancing.

\subsection{Problem Statement}
Analyzing queuing time with multiple parallel servers and general distribution of service time of jobs is known to be a hard problem and only approximations, based on first and second moment of service time distribution, have been proposed, e.g \cite{hokstad1978approximations,kimura1983diffusion}. Nevertheless, \cite{gupta2010inapproximability} shows that even these approximations are not accurate, since the queuing time in a queuing system varies with the third moment of service time distribution quite drastically \cite[Fig.1]{gupta2010inapproximability}. Queuing systems with redundancy, on the other hand, are studied in literature, e.g. \cite{gardner2017better,gardner2015reducing,raaijmakers2018delta,joshi2017efficient}. With redundancy, two scenarios for the cancellation of redundant copies of a jobs have been studied; cancellation after the first copy starts service, \cite{joshi2017efficient,raaijmakers2018delta}, and cancellation after the first copy finishes service, \cite{gardner2015reducing,gardner2017better}. As stated earlier, we consider the former scenario in this work.

The metric of interest for us is the queuing time in systems with redundancy. We define the queuing time of a job as the time interval between the arrival of the job and the entrance of its first copy to service. Analyzing the queuing time with the aforementioned model turns out to be a hard queuing problem \cite{raaijmakers2018delta}. To develop insight about this problem we propose the following analogy to the classical urns and balls problem.

\section{Performance Analysis}\label{urnsBallsAnalog}
\subsection{An Urns and Balls Analogy}
Let's consider the \textit{only-arrival} system, where servers do not run and jobs arrive and get scheduled to the queues of servers, according to an arbitrary scheduling policy. Note that we are interested in studying the performance of scheduling policies, in terms of queuing time. But as we mentioned earlier, the queuing time in a system with redundancy and bi-modal jobs' service time distribution, is determined by the scheduling policy's capability of providing load balancing and diversity of redundancy in the set of assigned servers across the jobs. We believe that a scheduling policy's capability of improving these indicators in the \textit{only-arrival} system is a good indicator of it's performance in the actual queuing system.

The characteristics of queues in the \textit{only-arrival} system can be studied with classical urns and balls problem, as follows. For $n$ urns and a given parameter $r$, such that $1\leq r\leq n$,  let's define $experiment 1$ as drawing $r$ urns, according to an arbitrary (scheduling) policy and without replacement from the set of $n$ urns, and placing a ball into each one. It is easy to see that the statistics of the occupancy of urns after repeating $experiment 1$ for $T$ times is same as the statistics of  the queues' length in the \textit{only-arrival} system after arriving $T$ jobs. In the rest of this paper we assume $n|T$. For a given set of parameters, i.e. $n$, $r$ and $T$, statistics of occupancy of the urns only depends on the (scheduling) policy, of how to select $r$ urns out of $n$. We use the term \textit{scheduling} for choosing urns due to its analogy to scheduling jobs into servers. In what follows, we will define the performance indicators for load balancing and diversity of redundancy. Then we will discuss and compare non-adaptive scheduling policies in terms of the defined indicators.

\subsection{Performance Indicators}
Let's define $N_i^T$, $i=1,2,\dots,n$, as the random variable for the number of balls in urn $i$ after $T$ repetitions of $experiment 1$. We also define the order statistics of $\{N_i\}_{i=1}^{n}$, as $N_{k:n}^{T}$, $k=1,2,\dots,n$.

\begin{definition}
The Load Balancing Factor (LBF) of a policy is defined as the ratio of the average number of balls in the minimum loaded urn to the average number of balls in the maximum loaded urn, when the policy is applied for drawing urns in $T$ repetitions of $experiment 1$,
    \begin{equation}
        \textup{LBF}_{policy}\coloneqq\frac{\mathbb{E}\left[N_{1:n}^T\right]}{\mathbb{E}\left[N_{n:n}^T\right]}.
    \label{def:LBF}
    \end{equation}{}
\end{definition}
Let's call that the set of chosen urns in the first repetition of $experiment1$ the \textit{initial set}. We define random variable $X$, as the number of overlaps in the set of chosen urns in an arbitrary repetition of $experiment 1$ with the initial set. The following definitions are based on the first and the second moment of this random variable. In the rest of this paper, we consider the case where $T\geq\lceil\frac{n}{r}\rceil$, in order for $X$ to have positive moments.

\begin{definition}
The Redundancy Overlap Factor (ROF) of a specific policy is defined as the inverse of the first moment of random variable X, 
    \begin{equation}
        \textup{ROF}_{policy}\coloneqq\frac{1}{\mathbb{E}[X]}.
    \label{def:ROF}
    \end{equation}{}
\end{definition}{}

\begin{definition}
The Redundancy Diversity Factor (RDF) of a specific policy is defined as the inverse of the second moment of positive random variable X,
    \begin{equation}
        \textup{RDF}_{policy}=\frac{1}{\mathbb{E}[X^2]}.
    \label{def:RDF}
    \end{equation}{}
\end{definition}{}

The importance of ROF and RDF can be explained as follows. The reason for introducing redundancy in a distributed system is to diversify the queues a job is assigned to, which in turn (on average) reduces the chance of all copies of the job to be stocked in a long queue. Therefore, especially in systems with bi-modal service time of jobs, overlaps between the set of (redundant) copies of consequent jobs could reduce or even nullify the benefit of redundancy. Hence, with lower average overlap, a policy makes a higher utilization of the redundant requests. On the other hand, between two policies with the same first moment, the one which has lower second moment is preferred, because the policy with larger second moment makes small overlaps in some repetitions and large overlaps in some others. With this behavior, the potential diversity that could be achieved by redundancy wont be leveraged. Thus, with keeping a moderate overlap in all repetitions of the $experiment1$, a policy with lower second moment gets the better out of the redundant requests. Next, we will study these indicators for different scheduling policies. Note that, greater values are desirable for all of LBF, ROF and RDF.

\subsection{Random Scheduling}
With this policy, each arriving job gets scheduled to $r$ servers, chosen uniformly at random without replacement, from the set of $n$ servers. under the urns and balls analogy, at each repetition of $experiment 1$, $r$ urns get selected randomly from the set of $n$ urns. With $r=1$, the problem boils down to the classical urns and balls problem where several analytical results exist for the occupancy of the urns, e.g. \cite{raab1998balls,dubhashi1998balls}. However, with $r>1$ there are very few works addressing the occupancy problem. Recently, the author of \cite{michelen2019short}, inspired by \cite{behrouzi2019average}, proposed asymptotic results for the number of balls in the maximum loaded bin, when $n\rightarrow\infty$. We adopted the following lemma from \cite{michelen2019short}.

\begin{lemma}
The load balancing factor, after $T$ repetitions of experiment 1, with random selection of urns can be asymptotically approximated by,
    \begin{equation}
        \underset{n\rightarrow\infty}{\textup{lim}}\textup{LBF}_{rand}=\textup{max}\left\{0,\frac{\frac{Tr}{n}-\sqrt{\frac{2Tr(n-r)\textup{log}(n)}{n^2}}}{\frac{Tr}{n}+\sqrt{\frac{2Tr(n-r)\textup{log}(n)}{n^2}}}\right\}.
    \end{equation}{}
\end{lemma}{}
\begin{proof}
The average number of balls in the maximum loaded bin, after $T$ repetitions of $experiment 1$ and with random scheduling, see \cite{behrouzi2019average}, is,
    \begin{equation}
        \mathbb{E}[N_{n:n}]=\frac{Tr}{n}+C_{n,r}\sqrt{T}+O\left(1/\sqrt{T}\right).
    \end{equation}{}
Then it is proved in \cite{michelen2019short} that,
    \begin{equation}
        \underset{n\rightarrow\infty}{\textup{lim}}C_{n,r}=\sqrt{\frac{2r(n-r)\textup{log}(n)}{n^2}}.
    \end{equation}{}
Therefore,
    \begin{equation}
        \underset{n\rightarrow\infty}{\textup{lim}}\mathbb{E}[N_{n:n}]=\frac{Tr}{n}+\sqrt{\frac{2Tr(n-r)\textup{log}(n)}{n^2}}.
    \end{equation}{}
Using the same approach as in \cite{michelen2019short}, the number of balls in the minimum loaded urn could be approximated by,
    \begin{equation}
        \underset{n\rightarrow\infty}{\textup{lim}}\mathbb{E}[N_{1:n}]=\frac{Tr}{n}-\sqrt{\frac{2Tr(n-r)\textup{log}(n)}{n^2}}.
        \label{minApprox}
    \end{equation}{}
However, since this approximations are derived using central limit theorem, negative values are also possible. In fact, it is easy to verify that for (\ref{minApprox}) to be positive, $T>2(\frac{n}{r}-1)\textup{log}n$. Therefore, an approximation of the number of balls in the minimum loaded urn could be given as,
    \begin{equation}
        \underset{n\rightarrow\infty}{\textup{lim}}\mathbb{E}[N_{1:n}]=\textup{max}\left\{0,\frac{Tr}{n}-\sqrt{\frac{2Tr(n-r)\textup{log}(n)}{n^2}}\right\},
    \end{equation}{}
which completes the proof.
\end{proof}
The approximations for the number of balls in the maximum/minimum loaded urns together with their experimental values are plotted in Fig. \ref{fig:max-min}. The approximated values follow the experiment very closely, both for the maximum loaded and the minimum loaded bins. In Fig. \ref{fig:lbf}, experimental and approximated values of LBF are plotted. The approximated values follow the experiment closely. Nevertheless, there is a gap between the two set of values, which decreases as $r$ gets closer to $n$. Note that, taking $r$ as a variable is only for showing the performance of our approximations and, later in this section, we will see that there should be a one-to-one relation between $r$ and $n$, for random scheduling to be comparable with the other policies.

    \begin{figure}[htbp]
        \centering
        \includegraphics[width=\columnwidth]{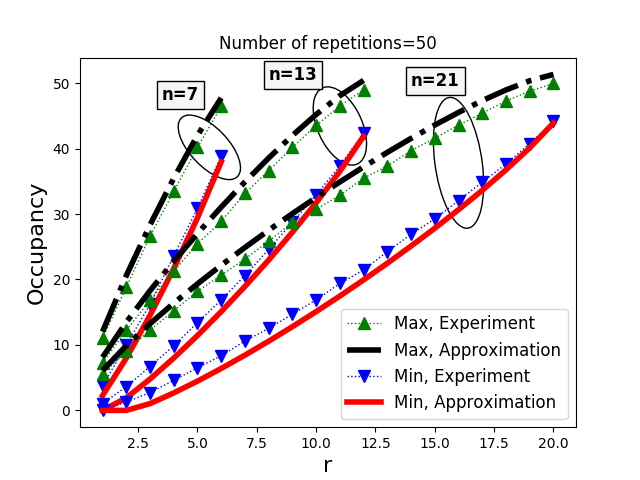}
        \caption{Average number of balls in the maximum and the minimum loaded urns after 50 repetitions of $experiment1$ (arrival of 50 jobs in the only-arrival system), as a function of the redundancy level $r$, for three different values of the number of urns (servers) $n$.}
        \label{fig:max-min}
    \end{figure}
    
    \begin{figure}[htbp]
        \centering
        \includegraphics[width=\columnwidth]{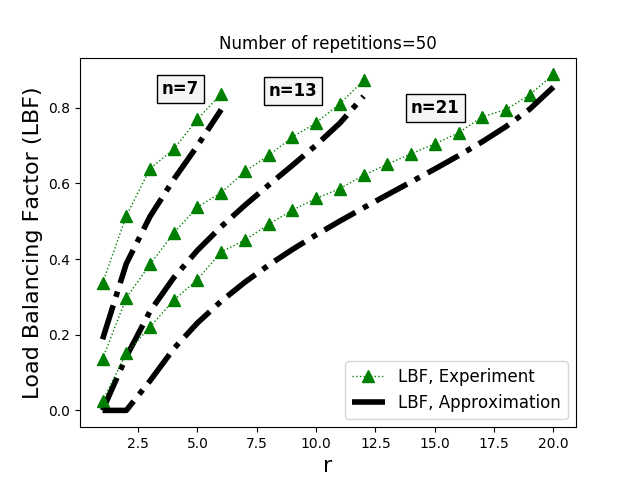}
        \caption{Load balancing factor (LBF) after 50 repetitions of $experiment1$ (arrival of 50 jobs in the only-arrival system), as a function of the redundancy level $r$, for three different values of the number of urns (servers) $n$.}
        \label{fig:lbf}
    \end{figure}
    
\begin{proposition}
With $T$ repetitions of experiment 1 and random urn selection in each round,
    \begin{equation}
            \textup{ROF}_{rand}=\frac{n}{r^2},~~~
            \textup{RDF}_{rand}=\frac{n(n-1)}{r^2(n+r(r-2))}.
    \label{RDFrandom}
    \end{equation}
\end{proposition}{}
The proof is omitted because of the space constraint.

\subsection{Round-Robin Scheduling}
With round-robin scheduling, the urns in each repetition of $experiment 1$ are chosen in a cyclic fashion. Particularly, for repetition $i$ the chosen urns are $((i-1)r+1)\mod{n},((i-1)r+2)\mod{n},\dots,((i-1)r+r)\mod{n}$, for $i\in[1,T]$.

\begin{proposition}
With round-robin scheduling,
    \begin{equation}
        \textup{LBF}_{round-robin}=1.
    \end{equation}{}
\end{proposition}{}
\begin{proposition}
With round-robin scheduling,
    \begin{equation}
        \begin{split}
            \textup{ROF}_{round-robin}&\approx\frac{n}{r^2}, \quad\textup{as}\quad T\rightarrow\infty,\\
            \textup{RDF}_{round-robin}&\approx\frac{3n}{2r^3+r}, \quad\textup{as}\quad T\rightarrow\infty.
        \end{split}{}
    \end{equation}{}
    \label{RDF:roundRobin}
\end{proposition}{}
\begin{proof}
With round-robin policy, let's consider (without loss of generality) that the initial set is $\{1,2,\dots,r\}$. Now suppose we repeat $experiment 1$ for $T-1$ times (after the first repetition), such that $n|T$. As the result, there will be $T-1$ sets of bins, chosen according to round-robin policy. It is easy to see that the number of sets that overlap with the initial set in $k$ place is $\frac{2T}{n}-1$, for $k\in[1,r-1]$, and $\frac{T}{n}-1$, for $k=r$. Therefore, the probability of a set of urns drawn in an arbitrary repetitions of $experiment 1$, within the following $T-1$ repetitions, to have exactly $k$ overlapping urns with the initial set is,
    \begin{equation}
        Pr\{X=k\} =
            \begin{cases}
                1-\frac{2r-1}{n}+\frac{r}{T} \quad & k=0,\\
                \frac{2}{n}-\frac{1}{T} \quad & k=1,2,\dots,r-1,\\
                \frac{1}{n}-\frac{1}{T} \quad & k=r,\\
                0 & \textup{otherwise}.
            \end{cases}       
    \end{equation}
Therefore,
    \begin{equation}
        \begin{split}
            \mathbb{E}[X]&=\left(\frac{1}{n}-\frac{1}{T}\right)r+\sum_{k=1}^{r-1}k\left(\frac{2}{n}-\frac{1}{T}\right)\\
            &\approx\frac{r^2}{n}\qquad\textup{as}\quad T\rightarrow\infty.
        \end{split}{}
        \label{EX_rr}
    \end{equation}
And
    \begin{equation}
        \begin{split}
            \mathbb{E}[X^2]&=\left(\frac{1}{n}-\frac{1}{T}\right)r^2+\sum_{k=1}^{r-1}k^2\left(\frac{2}{n}-\frac{1}{T}\right)\\
            &\approx\frac{2r^3+r}{3n}\qquad\textup{as}\quad T\rightarrow\infty.
        \label{EX2_rr}
        \end{split}{}
    \end{equation}
Respective substitution (\ref{EX_rr}) and (\ref{EX2_rr}) in (\ref{def:ROF}) and (\ref{def:RDF}), completes the proof.
\end{proof}{}

\subsection{Scheduling with Block Designs}
\subsubsection{Preliminary on Block Designs}
A block design is a pair $(\CMcal{X},\CMcal{A})$, where $\CMcal{X}$ is a set of objects and $\CMcal{A}$ is a set of non-empty subsets of $\CMcal{X}$, called \textit{blocks} \cite{stinson2007combinatorial}. In this work we focus on a certain type of designs called \textit{Balanced and Incomplete Block Designs} (BIBD), defined as follows. A $(\nu,k,\lambda)$-BIBD is a design that satisfies the following properties,
    \begin{enumerate}
        \item $|\CMcal{X}|=\nu$,
        \item every block consists of $k$ objects, and
        \item every pair of distinct objects is contained in exactly $\lambda$ blocks.
    \end{enumerate}{}
Note that, BIBD may not be possible for an arbitrary set of parameters. In fact, to satisfy the preceding properties the parameters are tightly dependent. Particularly in a symmetric BIBD, $\lambda(\nu-1)=k(k-1)$. In this work, we only consider symmetric BIBDs with $\lambda=1$. The reason for this choice is that with $\lambda=1$, the design provides minimum overlap between the blocks. As an example, with $\CMcal{X}=\{0,1,\dots,6\}$, the $(7,3,1)$-BIBD is
    \begin{equation*}
        012,034,056,146,157,247,256.
    \end{equation*}{}
For simplicity, we refer to a symmetric BIBD with $\lambda=1$ as BIBD in the rest of this paper.
\subsubsection{Scheduling with BIBD}
We choose $n$ and $r$ in $experiment 1$ such that there exist a $(n,r,1)$-BIBD. For that to exist, $n$ and $r$ should satisfy $n=r(r-1)+1$. With scheduling with BIBD, in each repetition of $experiment 1$, we choose $r$ urns the indices of which is dictated by one of the $n$ possible blocks, and blocks get selected in a round-robin fashion from the set of all blocks. For a visualization on the scheduling with BIBD visit \cite{catania}. According to the properties of the blocks in a BIBD, we next derive LBF, ROF and RDF for BIBD policy.
\begin{proposition}
With $T$ repetitions of $experiment 1$, such that $n|T$, and BIBD scheduling,
    \begin{equation}
        \textup{LBF}_{BIBD}=1.
    \end{equation}{}
\end{proposition}{}
\begin{proof}
It is easy to see that at the end of every $n$ rounds each urn have got selected $r$ times. Since $n|T$, each urn gets selected exactly $\frac{Tr}{n}$. Therefore urns get selected for equal number of times throughout the $T$ repetitions of $experiment 1$.
\end{proof}{}
\begin{proposition}
With BIBD scheduling, we have 
    \begin{equation}
            \textup{ROF}_{BIBD}\approx\frac{n}{n+r-1},~~
            \textup{RDF}_{BIBD}\approx\frac{n}{n+r^2-1}.
    \end{equation}
\end{proposition}{}
\begin{proof}
The proof is same in principle with the proof of Proposition \ref{RDF:roundRobin}. Except, the probability distribution for the random variable $X$ is,
    \begin{equation}
         Pr\{X=k\} =
            \begin{cases}
                \frac{n-1}{n} \quad & k=1,\\
                \frac{1}{n}-\frac{1}{T} \quad & k=r,\\
                0 & \textup{otherwise}.
            \end{cases}    
    \end{equation}{}
Therefore,
    \begin{equation}
        \begin{split}{}
            \mathbb{E}[X]&=\frac{n-1}{n}+\left(\frac{1}{n}-\frac{1}{T}\right)r\\
                 &\approx\frac{n-1}{n}+\frac{r}{n}\qquad\textup{as}\quad T\rightarrow\infty,\\
                 &=\frac{n+r^2-1}{n}.
        \end{split}
        \label{EX_bibd}
    \end{equation}{}
And,
    \begin{equation}
        \begin{split}{}
            \mathbb{E}[X^2]&=\frac{n-1}{n}+\left(\frac{1}{n}-\frac{1}{T}\right)r^2\\
                 &\approx\frac{n-1}{n}+\frac{r^2}{n}\qquad\textup{as}\quad T\rightarrow\infty,\\
                 &=\frac{n+r^2-1}{n}.
        \end{split}
        \label{EX2_bibd}
    \end{equation}
Respective substitution of (\ref{EX_bibd}) and (\ref{EX2_bibd}) in (\ref{def:ROF}) and (\ref{def:RDF}) completes the proof.
\end{proof}{}
\subsection{Comparison of the Indicators}
In order to compare the scheduling policies we have to maintain the relationship between $n$ and $r$ that is dictated by BIBD policy, i.e. $n=r(r-1)+1$.

\begin{table*}[t]
    \caption{Performance indicators of scheduling policies} 
    \centering 
    {\normalsize
    \begin{tabular}{l c c c} 
        \hline\hline  
        Policy & LBF($\leq 1$) & ROF($\leq 1$) & RDF($\leq 1$)
        \\ [0.5ex]
        \hline\hline 
        Random &  $\textup{max}\left\{0,\frac{\frac{Tr}{(r-1)^2+r}-\sqrt{\frac{2Tr(r-1)^2\textup{log}((r-1)^2+r)}{((r-1)^2+r)^2}}}{\frac{Tr}{(r-1)^2+r}+\sqrt{\frac{2Tr(r-1)^2\textup{log}((r-1)^2+r)}{((r-1)^2+r)^2}}}\right\}$ & $\frac{(r-1)^2+r}{r^2}$ & $\frac{(r-1)^2+r}{r(2r-1)}$\\
        Round-robin & 1 & $\frac{(r-1)^2+r}{r^2}$& $\frac{3[(r-1)^2+r]}{r(2r^2+1)}$\\
        BIBD & 1 & $\frac{(r-1)^2+r}{r^2}$ & $\frac{(r-1)^2+r}{r(2r-1)}$\\
        \hline 
    \end{tabular}}
    \label{tab:metrics}
\end{table*}

In terms of LBF, which we defined as the ratio of expected number of balls in the minimum loaded urn to that of the maximum loaded urn, round-robin and BIBD policies perform the same. In fact they achieve the highest possible load balancing, by assigning urns with equal number of balls when $n|T$. The load balancing property of random policy is inferior to the other policies, as it is shown in Fig. \ref{fig:lbfrdf}. On the other hand, in terms of ROF all three policies perform the same, as it can be seen in Table \ref{tab:metrics}. This is a surprising result, showing that if we select urns randomly at each repetition we would (on average) see the same amount of overlap as the case where the urns get selected with round-robin and BIBD, which are structured policies. In terms of RDF, random and BIBD policies perform equally, being superior to round-robin. The fact that random and BIBD have the same second moment of overlaps is quite counter-intuitive, since with random policy there is a higher number of possible sets, i.e. $\binom{n}{r}$, that can be selected at each repetition of $experiment 1$ compared to the number of possibilities with BIBD, i.e. $n$. Accordingly, one would expect the random policy to have a better performance in diversification of the sets of selected urns in the repetitions of $experiment 1$, compared to the (structured) BIBD policy. However, by looking closer to BIBD design it can be seen that, the gain of the policy, in terms of the number of non-interesting possible sets that it eliminates, is much greater than the pain of it, in terms of the number interesting possible sets that it prevents from happening. On the other hand, round-robin policy falls behind the other two policies in terms of RDF, which makes it inferior in handling the overlaps of the selected urns. Moreover, RDF of round-robin decreases as $r$ increases, which makes it a non-interesting policy when $r$ is large. In fact, we will see in Sec.~\ref{simResult} that the performance round-robin policy in a queuing system is more inferior to the other policies when $r$ is large. Putting all together, the average overlap is the same for all three policies, but in terms of diversity of redundancy round-robins is inferior to the other two policies.
    \begin{figure}[htbp]
        \centering
        \includegraphics[width=\columnwidth]{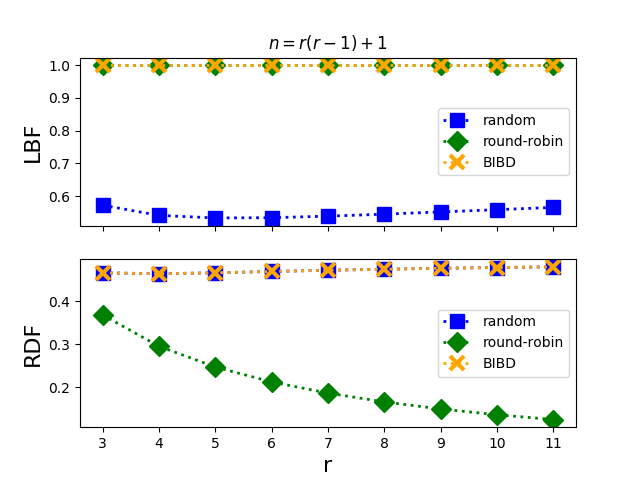}
        \caption{LBF and RDF for non-adaptive scheduling policies as a function of redundancy level $r$ and number of urns (servers) being $n=r(r-1)+1$. }
        \label{fig:lbfrdf}
    \end{figure}

In queuing systems, when jobs' service time are is bi-modal, the round-robin scheduling policy fails to manage the overlap between the assigned servers to subsequent jobs, increasing the chance of a small job getting trapped behind large jobs in many servers. On the other hand, random policy fails to maintain load balancing across the servers, which in turn increases the average waiting time in queues. Among the three non-adaptive policies, BIBD provides perfect load balancing and it has the best performance in terms of diversity of redundancy, which makes it a perfect candidate for scheduling of jobs in the system shown if Fig. \ref{fig:sysModel}. This observation will also be approved by simulations in the following section.

\section{Simulation Results}\label{simResult}
In this section we provide simulations results for queuing time with the scheduling policies discussed throughout the paper. We implemented a queuing system, with $n$ servers, each with its own queue, such that each arriving job gets scheduled to $r$ servers, selected according to a given scheduling policy. The redundant copies of a job get cancelled immediately after the first copy starts the service. If there is an idle server among the selected ones, the job enters service on that server, with no more redundant copies. With more than one idle server, ties are broken arbitrarily.

The bi-modal nature of jobs' service times has been observed in practice, e.g. Google Traces \cite{chen2010analysis} and is studied in queuing theory literature \cite{gupta2010inapproximability}. The service times of copies of a job are sampled from a fast exponential distribution if the job is short or from a slow exponential distribution if the job is long.
Assuming $1/\mu_1$ to be the average service time of the small jobs, using the model proposed in \cite{gardner2017better}, the average service time of large jobs is $1/\mu_2=q/\mu_1$, with $q$ being defined as (\ref{q}). Arrival of the jobs follows a Poisson process, with inter-arrival rate of $\lambda$. Each figure in this section consists of two sub-figures. The one at the left shows the queuing time for low arrival rates and the one at right shows the queuing time for high arrival rates. The set of parameters for each plot is shown by a 5-tuple, $(n,r,\mu_1,q,p)$.

Fig. \ref{fig:13_4_10_10} shows the average queuing time for the three scheduling policies, with $(n,r,\mu_1,q,p)=(13,4,10,10,0.1)$. In all arrival rates, BIBD policy outperforms both random and round-robin. The improvement ranges from about $10\%$ in high arrival rate up to $20\%$ in lower to medium arrival rates. The reason for less improvement in higher arrival regime is that, due to the higher average queue length, the contribution of the scheduling policy in the average waiting time is less dominant compared to the contribution of the queue length. With $(n,r,\mu_1,q,p)=(21,5,10,10,0.1)$, Fig. \ref{fig:21_5_10_10} shows that the relative performance of random and BIBD policies is almost the same as in $n=13$ and $r=4$. However, round robin policy becomes more inferior, due to larger $r$. The BIBD policy reduces the queuing time by $10\%$ and $25\%$ compared to random and round-robin, respectively. Note that, the larger the $r$ the higher the importance of handling overlaps between the redundant copies of jobs. As it is shown in Fig. \ref{fig:lbfrdf}, round-robin is the least effective policy for handling overlaps due to small RDF, which decreases as $r$ is increases.
    \begin{figure}[htbp]
        \centering
        \includegraphics[width=8.5cm,trim={1.8cm 0 2.7cm 0},clip]{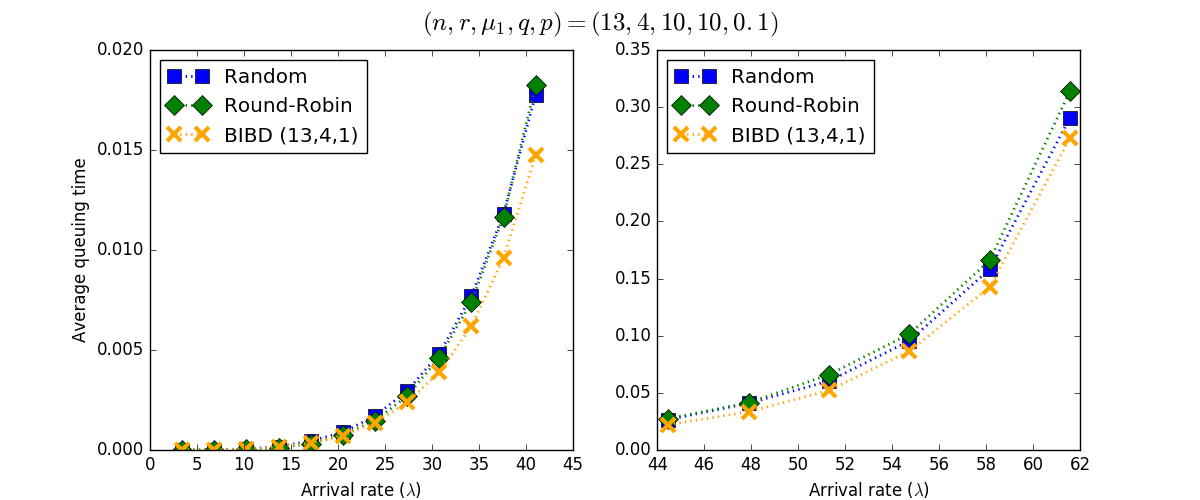}
        \caption{Average queuing time for non-adaptive scheduling policies, with $(n,r,\mu_1,q,p)=(13,4,10,10,0.1)$. }
        \label{fig:13_4_10_10}
    \end{figure}
    
    \begin{figure}[htbp]
        \centering
        \includegraphics[width=8.5cm,trim={1.8cm 0 2.7cm 0},clip]{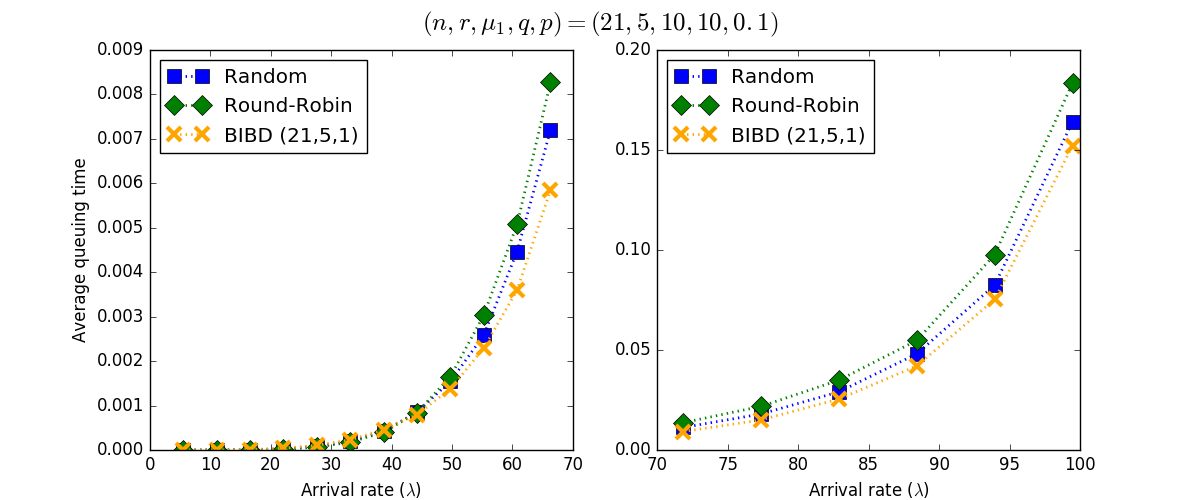}
        \caption{Average queuing time for non-adaptive scheduling policies, with $(n,r,\mu_1,q,p)=(21,5,10,10,0.1)$. }
        \label{fig:21_5_10_10}
    \end{figure}
Now if we keep the same $n$, $r$ and $p$ and just increase $q$, i.e. increase the average service time of the long jobs, with $(n,r,\mu_1,q,p)=(21,5,10,50,0.1)$, it can be seen from Fig. \ref{fig:21_5_50_10} that the performance gap between round robin, which does not handle overlaps, and BIBD/random gets larger. This observation holds for both low and high traffic regimes. With this set of parameters, BIBD reduces the queuing time by $100\%$ when compared to the round-robin policy. It is worth mentioning that, although round robin is an effective policy for load balancing, when the probability of large jobs is small but their average size is large, it fails to balance the loads on servers due to those large but not frequent jobs. Nevertheless, as it can be seen from Fig. \ref{fig:21_5_50_50}, when the probability of large jobs is high, with $(n,r,\mu_1,q,p)=(21,5,10,15,0.5)$, round robin performs closer to BIBD. On the other hand, random policy, which is more effective with handling overlaps, has a closer performance to BIBD when $p$ is smaller. This happens because, with small $p$ long jobs are not frequent and even BIBD, which balances the average load, fails in load balancing. Therefore, the load balancing capability of a scheduling policy has small impact on its performance when $p$ is small. In Fig. \ref{fig:21_5_50_50}, BIBD outperforms random and round-robin policies by up to $25\%$ and $50\%$, respectively.
    \begin{figure}[htbp]
        \centering
        \includegraphics[width=8.5cm,trim={1.8cm 0 2.7cm 0},clip]{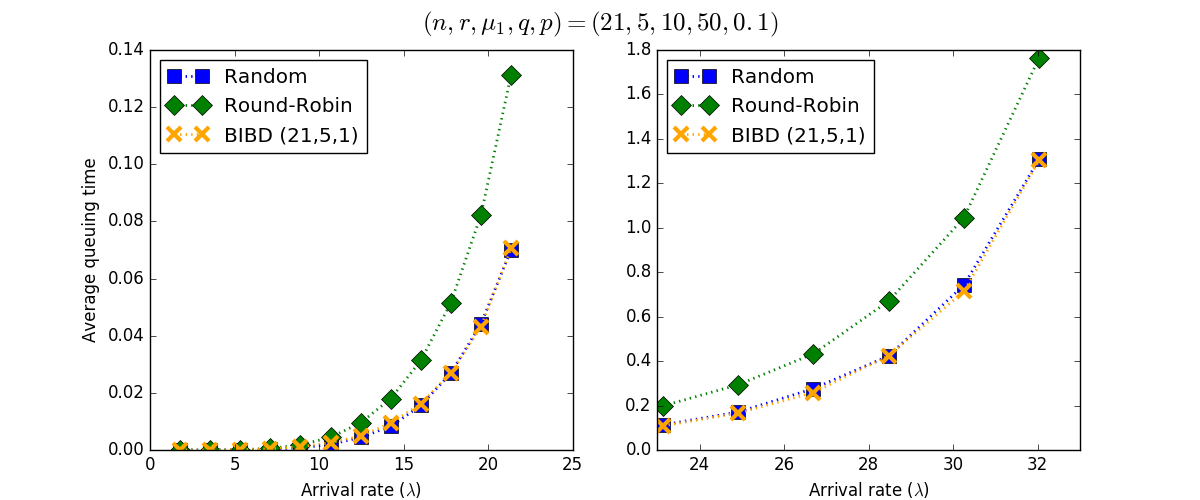}
        \caption{Average queuing time for non-adaptive scheduling policies, with $(n,r,\mu_1,q,p)=(21,5,10,50,0.1)$. }
        \label{fig:21_5_50_10}
    \end{figure}
    
    \begin{figure}[htbp]
        \centering
        \includegraphics[width=8.5cm,trim={1.8cm 0 2.7cm 0},clip]{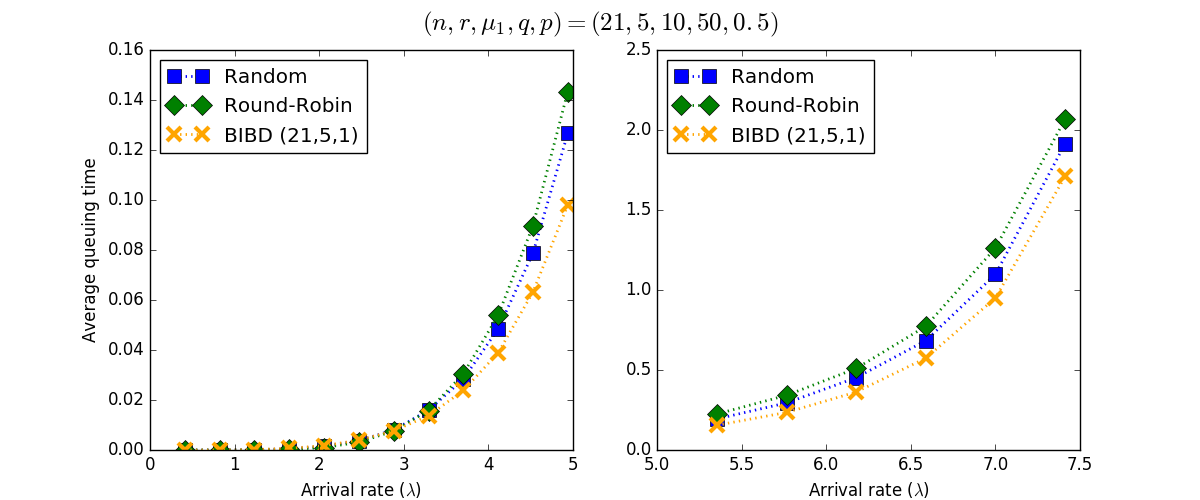}
        \caption{Average queuing time for non-adaptive scheduling policies, with $(n,r,\mu_1,q,p)=(21,5,10,50,0.5)$. }
        \label{fig:21_5_50_50}
    \end{figure}

\section{Conclusions and Future Work}\label{conclusion}
We studied the queuing time in a distributed system with $n$ servers, each with its own queue, where each arriving job gets replicated into $r$ servers, according to some non-adaptive scheduling policy. We considered jobs with bi-modal service times, where each arriving job has on average a long service time with probability $p$ and short service time with probability $1-p$. Specifically, the service time of each copy of short/long job was sampled from a fast/slow exponential distributions.. By developing an analogy to the classical urns and balls problem we introduced new performance indicators for scheduling policies. We studied random and round-robin scheduling policies in the analogous model and then proposed a new scheduling policy based on combinatorial block designs. We showed that the proposed policy provides higher diversity in the queuing time across the copies of a job and, therefore, leverages the potential of redundancy. Finally, by simulating a queuing system we showed that the proposed scheduling policy is indeed superior to the other policies, in terms of queuing time. As a future work, the steady state analysis of queuing systems with bi-modal job service time distribution can be considered. Using other types of block designs for load balancing and/or diversification of jobs' service time can also be the subject of future studies.
\section*{Acknowledgment}
Part of this research is based upon work supported by the NSF grants No.\ CIF-1717314 and CCF-1559855. The authors would like to thank Dr. Philip Whiting for his helpful comments and D.~Burke, E.~Catania, M.~Akta\c{s}, and P.~Peng for useful discussions.
\bibliographystyle{IEEEtran}
\bibliography{ref}

\end{document}